\documentclass[runningheads]{llncs}
\usepackage[utf8]{inputenc}
\usepackage{xspace}
\usepackage{amssymb,amsfonts,amsmath}
\usepackage{todonotes}
\usepackage{enumerate}
\usepackage{soul}
\usepackage{hyperref}
\usepackage[capitalise]{cleveref}
\usepackage{paralist}
\usepackage{lineno}
\usepackage{subfigure}
\usepackage[margin=2.5cm]{geometry}

\newtheorem{observation}{Observation}
\newtheorem{prop}{Property}

\let\doendproof\endproof
\renewcommand\endproof{~\hfill$\qed$\doendproof}

\newcommand{\verygood}{kite-planar $1$-planar\xspace}
\newcommand{\good}{$1$-kite-planar\xspace}
\newcommand{\piece}{piece of a kite\xspace}
\graphicspath{{images/}}
\title{On Morphing 1-Planar Drawings\thanks{to appear in Proc.\ of 47th International Workshop on Graph-Theoretic Concepts in Computer Science (WG 2021)}}
\author{Patrizio Angelini\inst{1}, Michael A. Bekos\inst{2},\\ Fabrizio Montecchiani\inst{3}, Maximilian Pfister\inst{2}}
\authorrunning{P.~Angelini, M.~Bekos, F.~Montecchiani, M.~Pfister}

\institute{%
John Cabot University, Rome, Italy\\
\email{pangelini@johncabot.edu}
\and
Institut f\"ur Informatik, Universit\"at T\"ubingen, T\"ubingen, Germany\\
\email{bekos@informatik.uni-tuebingen.de}, \email{maximilian.pfister@uni-tuebingen.de}
\and
Dipartimento di Ingegneria, Universit\`{a} degli Studi di Perugia, Perugia, Italy\\\email{fabrizio.montecchiani@unipg.it}
}

\begin{document}

\maketitle


\begin{abstract}
Computing a morph between two drawings of a graph is a classical problem  in computational geometry and  graph drawing.  
While this problem has been widely studied in the context of planar graphs,~very little is known about the existence of topology-preserving morphs for pairs of non-planar graph drawings. We make a step towards this problem by showing that a topology-preserving morph always exists for drawings of a meaningful family of $1$-planar graphs. While our proof is constructive, the vertices may follow trajectories of unbounded~complexity. 
\end{abstract}

\section{Introduction}
Computing a morph between two drawings of the same graph is a classical problem that attracted considerable attention over the years, also in view of its numerous applications in computer graphics and animations (refer to~\cite{AlamdariABCLBFH17} for a short overview). At high level, given two drawings $\Gamma_a(G)$ and $\Gamma_b(G)$ of the same graph $G$, a \emph{morph} between $\Gamma_a(G)$ and $\Gamma_b(G)$ is a continuously changing family of drawings such that the initial one coincides with $\Gamma_a(G)$ and the final one with $\Gamma_b(G)$. A standard assumption is that the two input drawings - as well as all intermediate ones - are \emph{topologically equivalent}, i.e., they define the same set of cells up to a homeomorphism of the plane (see \cref{sec:preliminaries} for formal definitions).
The main challenge is to design morphing algorithms that maintain some additional geometric properties of the input drawings throughout the transformation,  such as planarity with  straight-line edges (see, e.g.,~\cite{AlamdariABCLBFH17,cairns1944,fg-hmti-99}),  convexity~\cite{DBLP:conf/compgeom/AngeliniLFLPR15,thomassen1983}, orthogonality~\cite{Biedl:2013:MOP:2533288.2500118,DBLP:conf/gd/GoethemSV19,DBLP:conf/gd/GoethemSV19}, and upwardness~\cite{ddfpr-upm-gd-18}. We point the interested reader to~\cite{DBLP:conf/compgeom/AngeliniCCLR19,abcdd-pd3-gd-18,DBLP:conf/wads/Barrera-CruzBLB19,barreraCruz/haxel/lubiw:18} for additional related work.

In this context, the most prominent research direction focuses on morphs of straight-line planar drawings: The topological equivalence condition implies that  all drawings in the morph have the same planar embedding; in addition, it is also required that edges remain straight-line segments. Back in 1944,  Cairns~\cite{cairns1944} proved that such morphs always exist if the input graphs are plane triangulations. This implies that, for a fixed plane triangulation, the space of its straight-line planar drawings is connected. The main drawback of Cairns result is in the underlying construction, which involves exponentially-many morphing steps. 
The extension of Cairns' result to all plane graphs was initially done by  Thomassen~\cite{thomassen1983}, while later  Floater and Gotsman~\cite{fg-hmti-99}, and  Gotsman and Surazhsky~\cite{DBLP:journals/cg/GotsmanS01,DBLP:journals/tog/SurazhskyG01} proposed different approaches using trajectories of unbounded complexity. More recently, Alamdari et al.~\cite{DBLP:conf/soda/AlamdariACBFLPRSW13} focused on the complexity of the morph. They described the first morphing algorithm for planar straight-line drawings that makes use of a polynomial number of steps, where in each step vertices move at uniform speed along linear trajectories. In a subsequent paper~\cite{AlamdariABCLBFH17}, a linear bound on the number of steps is shown, which is worst-case optimal.

Morphing non-planar  drawings of graphs appears to be a more elusive problem. In particular, Angelini et al.~\cite{DBLP:conf/icalp/AngeliniLBFPR14} asked whether a morphing
algorithm exists for pairs of non-planar straight-line drawings such that the topology of the crossings in the drawing is maintained throughout the morph. They stressed that a solution to this problem is not known even if the vertex trajectories are allowed to have arbitrary complexity. Note that the obvious idea of morphing the ``planarizations'' of the drawings (i.e., the planar drawings obtained by treating crossings as dummy vertices) does not trivially work. Namely, in order to guarantee that edges remain straight-line segments throughout the morph, one has to ensure that opposite edges incident to dummy vertices maintain the same slope. To the best of our knowledge, such requirement cannot be easily incorporated into any of the already known morphing algorithms for planar graphs. 

One way of simplifying the problem is to consider graphs that are non-planar but still admit embeddings on surfaces of bounded genus. In this direction, Chambers et al.~\cite{morphing-torus:soda2021} proved the existence of morphs for pairs of  crossing-free drawings on the Euclidean flat torus (where edges are still geodesics). Their technique is rather complex and the authors concluded that an extension to higher genus surfaces is fairly non-trivial. 

We make a step towards settling the open problem in~\cite{DBLP:conf/icalp/AngeliniLBFPR14} by studying non-planar drawings of graphs with forbidden edge-crossing patterns. Our focus is on the family of \emph{$1$-planar} graphs, which naturally extends the notion of planarity by allowing each edge to be crossed at most once (see~\cite{DBLP:journals/csr/KobourovLM17} for a survey). Note that $1$-planar graphs form a well studied family of non-planar graphs with early results dating back to the 60's~\cite{avital-66,MR0187232}, while more recently they have gained considerable attention in the rapidly growing literature about beyond planarity~\cite{DBLP:journals/csur/DidimoLM19,DBLP:books/sp/20/HT2020}.

\begin{figure}[t]
    \centering
    \subfigure{\includegraphics[width=0.35\textwidth,page=1]{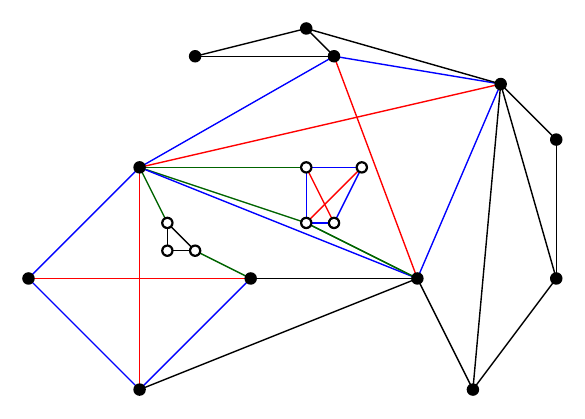}}\hfil
    \subfigure{\includegraphics[width=0.35\textwidth,page=2]{good}}
    \caption{Two topologically-equivalent  \verygood drawings of the same graph.\label{fig:verygood}}
\end{figure}
\smallskip\noindent\emph{Our contribution.} We provide a set of sufficient conditions under which any pair of $1$-planar straight-line drawings admits a morph. At high-level, we require that if two edges cross, then they can be enclosed in a quadrilateral region whose boundary is uncrossed; although this region may contain further vertices in its interior, we require that any edge connecting an end-vertex of the crossing edges to a vertex inside the region is also uncrossed; refer to \cref{fig:verygood} for an illustration. A drawing that satisfies these requirements is called \emph{\verygood} (see \cref{def:good-drawing}). Our main result is summarized by the following theorem.

\begin{theorem}\label{thm:main}
There exists a morph between any pair of topologically-equivalent \verygood drawings.
\end{theorem}

\noindent \cref{thm:main} implies that, for a fixed graph, the space of its topologically-equivalent \verygood drawings is connected. The proof is constructive, although the vertices may use trajectories of unbounded complexity. 
Concerning the definition of \verygood drawings, it may be worth observing that, due to a simple edge density argument, the graphs admitting such a drawing cannot be embedded on any surface of bounded genus. Indeed, as shown in \cref{se:implications}, some well-known families of $1$-planar graphs 
admit drawings that are \verygood and require arbitrary large genus to be embedded. 

\smallskip\noindent\emph{Paper structure.} \cref{sec:preliminaries} contains basic definitions and notation. \cref{sec:main} gives an overview of the proof technique, which exploits a recursive construction. The base case of the recursion is described in \cref{sec:base}, while the recursive step is in \cref{sec:recursive}. Implications of our result in terms of classes of $1$-planar drawings that admit a morph are discussed in    \cref{se:implications}. Open problems are given in \cref{sec:conclusions}.

\section{Preliminaries}\label{sec:preliminaries}

\noindent\textbf{Drawings.} A \emph{straight-line drawing} (or simply a \emph{drawing}, for short) $\Gamma(G)$ of a graph $G$ maps each vertex $v$ of $G$ to a distinct point $p_v$ of the plane and each edge $(u,v)$ of $G$ to a straight-line segment connecting $p_u$ and $p_v$ without passing through any other point representing a vertex of $G$. When this creates no ambiguities, we will not distinguish between a vertex and the point representing it in $\Gamma(G)$, as well as between an edge and its segment. Note that, by definition, two edges of a drawing share at most one point, which is either a common endpoint or an interior point where the two edges properly cross. Drawing $\Gamma(G)$ partitions the plane into connected regions called \emph{cells}. The boundary of a cell consists of vertices, crossing points, and (parts of) edges. The \emph{external cell} of $\Gamma(G)$ is its (only) unbounded cell.  Two drawings $\Gamma_a(G)$ and $\Gamma_b(G)$ of the same graph $G$ are \emph{topologically equivalent} if they define the same set of cells up to a homeomorphism of the plane.  An \emph{embedding} of $G$ is an equivalence class of drawings that are pairwise topologically equivalent. 

A drawing $\Gamma(G)$ is \emph{planar} if no two edges cross. In this case, the cells of $\Gamma(G)$ are called \emph{faces} and their boundaries consist of just vertices and edges. A graph is \emph{planar} if it admits a planar drawing. A planar graph  together with an embedding defined by a planar drawing is a \emph{plane graph}. A planar drawing is \emph{strictly convex} if all its faces are strictly convex polygons. 

A graph is \emph{$1$-planar} if it admits a (not necessarily straight-line) \emph{$1$-planar drawing} in which every edge crosses at most one other edge. A $1$-planar graph together with an embedding defined by a $1$-planar drawing is a \emph{$1$-plane graph}. A \emph{kite} $K$ in a $1$-planar drawing $\Gamma(G)$ is a $1$-planar drawing of $K_4$ in $\Gamma(G)$ whose external cell is a quadrilateral.  The four edges on the boundary of the external cell of $K$ are called \emph{kite edges}. The other two edges are the \emph{crossing edges} of $K$ and are drawn inside the quadrilateral bounding $K$. A \emph{partial kite} is a subdrawing of a kite such that some kite edges may not exist. \Cref{fig:verygood} shows two kites and one partial kite; the kite (crossing) edges are blue (red, resp.). 

Throughout the paper, we assume that the four kite-edges of any (partial) kite exist and are uncrossed; this will be justified by P.\ref{pr:p2} of \cref{def:good-drawing}. In other words, we assume that, in the considered drawings, no kite is partial. Given a vertex $v$ of $G$ and a kite $K$, the following exclusive cases can occur:
\begin{inparaenum}[(i)]
\item $v$ \emph{belongs} to $K$, if it is a vertex of the $K_4$ defining $K$, or
\item  $v$ is \emph{inside} $K$ (or $K$ \emph{contains} $v$) if $v$ lies in the interior of the quadrilateral bounding $K$, or
\item $v$ is outside $K$, otherwise.
\end{inparaenum}
A kite is \emph{empty} if it contains no vertex; otherwise, it is \emph{non-empty}. An edge $(u,v)$ is a \emph{binding edge} (green in \cref{fig:verygood}) if $u$ belongs to a non-empty kite $K$ and $v$ is inside $K$. 
We can now introduce \verygood drawings.

\begin{definition}\label{def:good-drawing}
A straight-line drawing is \emph{\verygood}, or \emph{\good} for short, if:
\begin{inparaenum}[\bf(P.1)]
\item\label{pr:p1}~every edge is crossed at most once,
\item\label{pr:p2}~the four kite~edges of every (possibly partial) kite are either present or can be drawn with uncrossed straight-line segments, and
\item\label{pr:p3}~every binding edge is uncrossed.
\end{inparaenum}
\end{definition}

Let $\Gamma(G)$ be a \good drawing of $G$. We say that a vertex of $G$ is of \emph{level $0$} if no kite contains it, while it is of \emph{level $i>0$} if the maximum level of the vertices belonging to a kite containing it is $i-1$. In \cref{fig:verygood}, the black (white) vertices are of level $0$ (level $1$, resp.). The next property follows from  P.\ref{pr:p3} of \Cref{def:good-drawing}.
\begin{prop}\label{pr:kite-same-level}
If two vertices belong to the same kite of a \good drawing $\Gamma(G)$, then they are of the same level.
\end{prop}
\begin{proof}
Suppose, for a contradiction, that two vertices $v$ and $w$ that belong to a kite $K$ are such that the level of $v$ is larger than the level of $w$. This implies that $v$ is inside a kite $K'$ that does not contain $w$. By 1-planarity, $w$ is a kite-vertex of $K'$. It follows that there is a crossing edge of $K$ incident to $w$ that is a binding edge with respect to $K'$, which contradicts Property P.\ref{pr:p3} of \Cref{def:good-drawing}.
\end{proof}

\smallskip\noindent\textbf{Morphs.} Let $\Gamma_a(G)$ and $\Gamma_b(G)$ be two topologically-equivalent drawings of the same graph $G$. A \emph{morph} between them is  a continuously changing family of pairwise topologically-equivalent drawings of $G$ indexed by time $t \in  [0, 1]$, such that the drawing at time $t = 0$ is $\Gamma_a(G)$ and the drawing at time $t = 1$ is $\Gamma_b(G)$. Since edges are drawn as straight-line segments, a morph is uniquely specified by the vertex trajectories. Also, during the course of the morph, a vertex may coincide with neither another vertex nor an internal point of an edge.

\section{Outline of the Proof of \Cref{thm:main}}\label{sec:main}

In this section, we give an outline of the proof of \Cref{thm:main}, namely, that there exists a morph between any two topologically-equivalent \good drawings $\Gamma_a(G)$ and $\Gamma_b(G)$ of a graph $G$. Recall that $\Gamma_a(G)$ and $\Gamma_b(G)$ define the same embedding of $G$. Hence, $G$ is necessarily a $1$-plane graph. 

Our proof is by means of a recursive construction. The underlying idea is to compute a morph by keeping each kite boundary drawn as a strictly-convex polygon, so that, in the course of the morph, the drawing of the corresponding crossing edges will stay inside their boundary. The main challenge, however, stems from the fact that a kite may not be empty. Therefore, our approach is to remove the interior of each kite, recursively compute a morph that keeps the convexity of the kite boundaries, and suitably reinsert (and morph) the removed subdrawings.  
In the proof, we will use two key ingredients.  The first one is a result by Aronov et al.~\cite{ass-ctsp-93}, which guarantees that one can \emph{compatibly triangulate} two topologically-equivalent planar drawings of a planar graph.

\begin{theorem}[Aronov et al.~\cite{ass-ctsp-93}]\label{th:compatible}
Given two topologically-equivalent planar drawings $\Gamma_a(P)$ and $\Gamma_b(P)$ of the same $n$-vertex planar graph $P$, it is possible to augment $\Gamma_a(P)$ and $\Gamma_b(P)$ to two topologically-equivalent planar drawings $\Gamma_a(P')$ and $\Gamma_b(P')$ of the same maximal planar graph $P'$ such that $\Gamma_a(P) \subseteq \Gamma_a(P')$, $\Gamma_b(P) \subseteq \Gamma_b(P')$, and the order of $P' \setminus P$ is $O(n^2)$.
\end{theorem}

The second ingredient is a result by Angelini et al.~\cite{DBLP:conf/compgeom/AngeliniLFLPR15}, which allows us to morph a pair of convex drawings by preserving the convexity of the faces. The main properties of this result are summarized in the next theorem.

\begin{theorem}[Angelini et al.~\cite{DBLP:conf/compgeom/AngeliniLFLPR15}]\label{th:convex}
Let $\langle \Gamma_a(P), \Gamma_b(P) \rangle$ be a pair of topologically-equivalent strictly-convex planar drawings of a graph $P$. There is a morph between $\Gamma_a(P)$ and $\Gamma_b(P)$ in which every intermediate drawing is strictly convex. If the outer face of $G$ has only three vertices and each of them has the same position in $\Gamma_a(P)$ and $\Gamma_b(P)$, then these three vertices do not move during~this~morph.
\end{theorem}
We apply recursion on the maximum level $\ell$ of a vertex of $G$. 
The base case ($\ell=0$) is described in \cref{sec:base}, while the recursive case ($\ell>0$) in \cref{sec:recursive}. \Cref{fig:schema} provides a high-level description of the~main steps~in~the~proof.

\begin{figure}[b]
    \centering
    \includegraphics[scale=0.5,page=2]{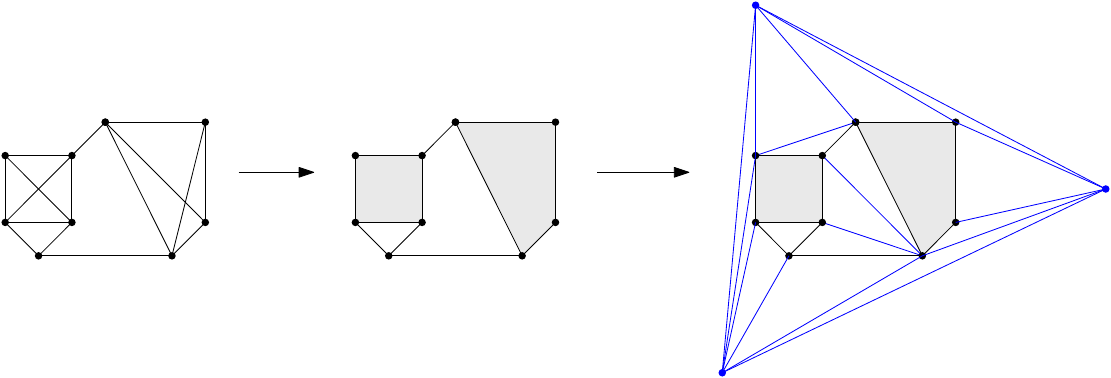}
    \caption{Illustration of the transitions $\Gamma_a(G) \rightarrow \Gamma_a(P) \rightarrow \Gamma_a(P')$; marked faces are gray.\label{fig:base-case}}
\end{figure} 

\begin{figure}[p]
\includegraphics[width=0.9\textwidth]{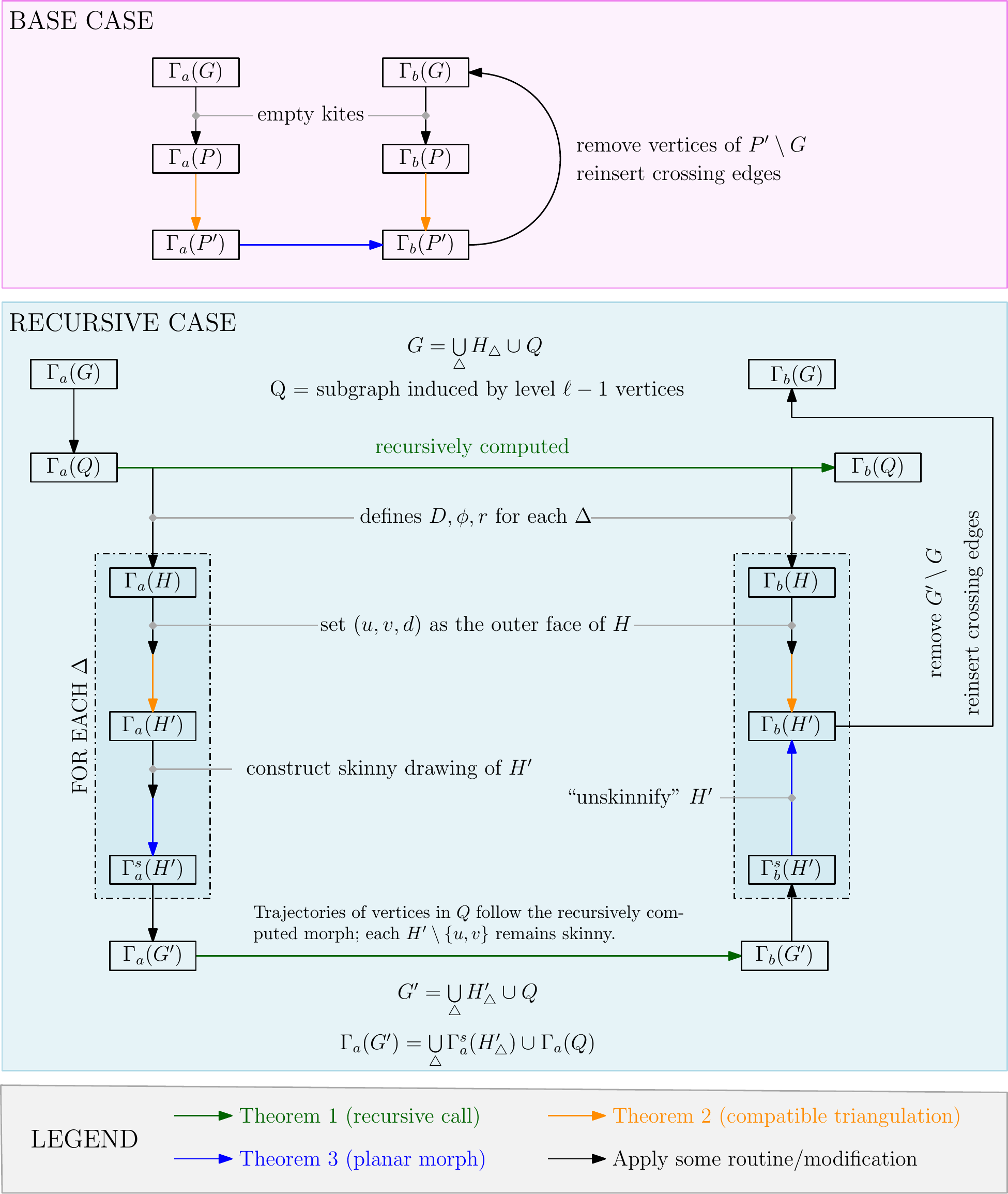}
\caption{Schematic illustration for the proof of \cref{thm:main}\label{fig:schema}}
\end{figure}

\section{Base case}\label{sec:base}

In the base case of the recursion, all the vertices of $G$ are of level $0$, which implies that all the kites of $G$, if any, are empty. 
Let $P$ be the graph obtained by removing both crossing edges from each kite of $G$. Let $\langle \Gamma_a(P), \Gamma_b(P) \rangle$ be the restrictions of $\langle \Gamma_a(G), \Gamma_b(G) \rangle$ to $P$, respectively; see \cref{fig:base-case}. By construction, $\langle \Gamma_a(P), \Gamma_b(P) \rangle$ is a pair of planar and topologically-equivalent drawings, and $P$ is a plane subgraph of $G$. The kite edges of each kite $K$ of $G$ are uncrossed (by P.\ref{pr:p2} of \cref{def:good-drawing}) and bound a quadrangular face $f_K$ in $P$, which we~call~\emph{marked}.

Let $P'$ and $\langle \Gamma_a(P'), \Gamma_b(P') \rangle$ be the graph and the corresponding pair of planar drawings obtained by applying \Cref{th:compatible} to $\langle \Gamma_a(P), \Gamma_b(P) \rangle$, except for the marked faces; see Fig.~\ref{fig:base-case} for an illustration. This operation guarantees that 
every face in both drawings is a triangle, if not marked, or a quadrangle, if marked. We call a plane graph with such faces \emph{almost triangulated}, and we next prove that it is triconnected.


\begin{lemma}\label{le:almost}
Every almost triangulated plane graph is triconnected.
\end{lemma}
\begin{proof}
Let $P'$ be an almost triangulated plane graph derived from a \good drawing of a $1$-planar graph $G$. Suppose that $P'$ contains a separation pair $\{u,v\}$. Then there exist at least two faces $f_1$ and $f_2$ that are incident to both $u$ and $v$ such that at least one, say $f_2$, is not triangular and hence is marked. In the presence of this marked face, edge $(u,v)$ exists in $G$ and not in $P'$. Consequently, $f_1$ cannot be a triangle, as otherwise it would contain edge $(u,v)$ on its boundary. On the other hand, if $f_1$ is marked,  then $G$ contains another copy of $(u,v)$ drawn inside the kite that yielded $f_1$, which is impossible since $G$ is simple. Hence, $P'$ contains no separation pair. The absence of cutvertices stems from the fact that each face is either a triangle or a quadrangle (if marked).
\end{proof}

Since each kite contains two crossing edges in $G$, its boundary is drawn~strictly convex in both $\Gamma_a(G)$ and  $\Gamma_b(G)$. Hence, $\Gamma_a(P')$ and $\Gamma_b(P')$ are two strictly convex planar drawings of $P'$. This property allows to apply \Cref{th:convex} to compute a morph of  $\langle \Gamma_a(P'), \Gamma_b(P') \rangle$  that maintains the strict convexity of the drawing at any time instant. 
Since each marked face $f_K$ remains strictly convex, adding back the two crossing edges of the corresponding kite $K$ in $P'$ yields a morph of a supergraph of $G$ (and thus of $G$) in which these crossing edges remain inside the boundary of $K$ at any time instant. This concludes the base case.

\section{Recursive case}\label{sec:recursive}

In this section, we focus on the recursive step of the proof of \Cref{thm:main}, in which the maximum level of a vertex in $G$ is $\ell > 0$. 
Let $Q$ be the graph obtained by removing all the vertices of level $\ell$ from $G$, and let $\langle \Gamma_a(Q),\Gamma_b(Q) \rangle$ be the restriction of $\langle \Gamma_a(G),\Gamma_b(G)\rangle$ to $Q$. Clearly, the two drawings of $Q$ are topologically equivalent and the maximum level of a vertex is $\ell-1$.  Thus, we can recursively compute a morph of $\langle \Gamma_a(Q),\Gamma_b(Q) \rangle$. In what follows, we describe how to incorporate the trajectories of the level-$\ell$ vertices into the morph of $\langle \Gamma_a(Q),\Gamma_b(Q) \rangle$, so to obtain the desired morph of $\langle \Gamma_a(G),\Gamma_b(G)\rangle$. 

\medskip\noindent\textbf{Setting up the morph.} 
We begin by observing that, by \Cref{pr:kite-same-level}, each vertex of level $\ell$ is contained in a kite whose vertices are all of level $\ell-1$, which implies that this kite is empty in $Q$ (but not in $G$). Consider such a kite $K$, and note that its two crossing edges define four triangular regions that remain non-degenerate during the morph of $\langle \Gamma_a(Q),\Gamma_b(Q)\rangle$. We refer to each of these four triangular regions as a \emph{\piece}. Consider a piece of kite $K$ and denote it by $\triangle$. The unique edge $(u,v)$ of $\triangle$ that belongs to the boundary of $K$ is called the \emph{base edge} of $\triangle$. Since $\triangle$ remains non-degenerate during the morph, there exists a half-disk $D$ that, throughout the whole morph, has the following properties (see also \cref{fig:disk} for an illustration):

\begin{itemize}
    \item half-disk $D$ lies in $\triangle$ and is centered at the midpoint $w$ of $(u,v)$, and
    \item the length of its radius is positive and it does not change.
\end{itemize}

\begin{figure}[h]
    \centering
    \includegraphics[scale=0.9,page=1]{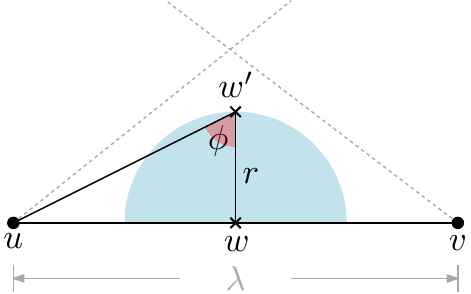}
    \caption{Illustration of the half-disk $D$ of $\triangle$ and their geometric properties.\label{fig:disk}}
\end{figure}

\noindent Let $\lambda$ be the smallest length of the base edge $(u,v)$ during the morph, let $r$ be the radius of $D$ perpendicular to $(u,v)$, and let $w'$ be the endpoint of $r$ different from $w$. Also, denote by $t^*$ any time instant of the morph when the length of $(u,v)$ equals $\lambda$, and let $\phi$ be the internal angle at $w'$ of the triangle formed by $u,w$ and $w'$ at time $t^*$. In particular, $\phi$ satisfies $\tan(\phi) = \frac{\lambda}{2} \cdot \frac{1}{|r|}$.

Consider the graph $\mathcal{H} = G \setminus Q$ induced by the level-$\ell$ vertices of $G$, and let $H_{\triangle}$ be the subgraph of $\cal H$ that lies inside $\triangle$. We proceed to compute a drawing of  $H_{\triangle}$ that, intuitively, will be ``small'' enough to fit inside $D$ and ``skinny'' enough to avoid crossings with the binding edges that connect $u$ or $v$ to $H_{\triangle}$. 
To ease the notation, from now on we will refer to $H_{\triangle}$ as $H$.

To compute this drawing, we first augment $H$ as well as its drawings in $\langle \Gamma_a(G), \Gamma_b(G) \rangle$, as follows. We add a dummy vertex $d$ connected to $u$ and to $v$, which is drawn sufficiently close to the crossing point of the two diagonals of $K$ in both $\Gamma_a(G)$ and $\Gamma_b(G)$, so that the triangle formed by $u$, $v$, and $d$ contains $H$.

As in the transition from $P$ to $P'$ in \cref{sec:base}, we remove the crossing edges of every (empty) kite of $H \cup \{u,v,d\}$ and we mark the resulting quadrangular face. Then we apply \cref{th:compatible} to the resulting planar subgraph of $H \cup \{u,v,d\}$ and to its drawings in $\langle \Gamma_a(G), \Gamma_b(G) \rangle$, except for its marked faces. This results in an almost triangulated plane graph $H'$ and in a pair of topologically-equivalent strictly convex drawings $\langle \Gamma_a(H'),\Gamma_b(H') \rangle$ of $H'$. The following observation is directly implied by Property P.\ref{pr:p3} of \cref{def:good-drawing}.

\begin{observation}\label{obs:triangles-at-u}
Every face incident to $u$ or to $v$ in $H'$ is triangular.
\end{observation}
Consider the plane graph obtained from $H'$ by removing $u$ and $v$, and let $\cal C$ be the graph formed by the vertices and the edges of its outer face. In the following lemma, we investigate some properties of $\cal C$. The \emph{BC-tree} $\mathcal{T}$ of a connected graph $G$ represents the decomposition of $G$ into its biconnected components, called \emph{blocks}. Namely, $\mathcal{T}$ has a \emph{B-node} for each block of $G$ and a \emph{C-node} for each cutvertex of $G$, such that each B-node is connected to the C-nodes that are part of its block.

\begin{lemma}\label{le:chain}
The following properties of $\cal C$ hold:
\begin{enumerate}[a)]
    \item \label{ch:1} $\cal C$ is outerplane and connected.
    \item \label{ch:2} Each block of $\cal C$ is a cycle, possibly degenerated to a single edge.
    \item \label{ch:3} Every cutvertex of $\cal C$ is connected to both $u$ and $v$ in $H'$.
    \item \label{ch:4} The $BC$-tree of $\cal C$ is a path.
    \item \label{ch:5} Every non cutvertex of $\cal C$ is connected to exactly one of $u$ and $v$ in $H'$, with the exception of exactly two vertices (one of them is $d$) which belong to the blocks of $\cal C$ corresponding to degree-$1$ $B$-nodes in the $BC$-tree of $\cal C$. 
\end{enumerate}
\end{lemma}
\begin{proof}
Since $\cal C$ is formed by the vertices and edges of the outer face of $H' \setminus \{u,v\}$, $\cal C$ is outerplane. Also, since $H'$ is triconnected, by \cref{le:almost}, removing two vertices, namely $u$ and $v$, cannot yield a disconnected graph. This proves Property~$(\ref{ch:1})$.
Property $(\ref{ch:2})$ directly follows from the fact that $\cal C$ is formed by the vertices and edges of the outer face of $H' \setminus \{u,v\}$.

Concerning Property~$(\ref{ch:3})$, we actually prove a stronger variant. Let $c$ be a cutvertex of $\cal C$, and consider a closed walk along the boundary of the outer face of $\cal C$. We claim that each occurrence of $c$ along this walk implies the existence, in $H'$, of one copy of edge $(u,c)$ or of one copy of edge $(v,c)$. Since $H'$ is simple, with this claim at hand, it follows that $c$ occurs exactly twice along the walk, and that $c$ is adjacent to both $u$ and $v$. To prove the claim, let $a$ and $b$ be the vertex before and the vertex after an occurrence of $c$ along the walk, respectively. Namely, $a$ and $b$ belong to different biconnected components sharing only $c$. We further distinguish two cases:
\begin{itemize}
    \item There is a face of $H'$ that contains both $a$ and $b$. Since $H'$ is almost triangulated, such face is either a triangle or a marked quadrangle. In the former case, the edge $(a,b)$ would be part of $\cal C$, which is impossible, because $c$ is a cutvertex separating $a$ and $b$. In the latter case, the marked face contains $a,b,c$ and a fourth vertex $x$. By \cref{obs:triangles-at-u}, $x$ is neither $u$ nor $v$, which implies that $x$ belongs to $\cal C$. This is again impossible, because $c$ is a cutvertex separating $a$ and $b$. 
    \item There is no face of $H'$ that contains both $a$ and $b$. In this case, there exists an edge $(x,c)$ belonging to $H'$, such that $(x,c)$ does not belong to $\cal C$. Hence, $x$ is either $u$ or $v$, as desired.
\end{itemize}

To show that the $BC$-tree of $\cal C$ is a path (Property~$(\ref{ch:4})$), i.e., all the $B$- and $C$-nodes of the $BC$-tree have degree at most $2$, we prove equivalently that each cutvertex of $\cal C$ is shared by exactly two blocks and each block contains at most two cutvertices. The former follows from the fact that each cutvertex of $\cal C$ occurs exactly twice in a closed walk along the boundary of the outer face of $\cal C$ by Property $(\ref{ch:3})$. For the latter, we observe that if a block contains three cutvertices, $c_1$, $c_2$ and $c_3$, then the boundary of $\cal C$ contains three vertex-disjoint paths that pairwise connect these cutvertices. Since $u$ and $v$ are connected to each other in $H'$ and since by Property $(\ref{ch:3})$ $u$ and $v$ are connected to each of $c_1$, $c_2$ and $c_3$, $H'$ contains $K_5$ as a minor, which is impossible because $H'$ is planar. 

Concerning Property~$(\ref{ch:5})$, we first show the existence of the two non-cut vertices of $\cal C$ that are connected to both $u$ and $v$. The first one is vertex $d$ by construction. For the other one consider the internal face of $H'$ that is incident to the edge $(u,v)$. By \cref{obs:triangles-at-u}, this face is not marked and thus the third vertex incident to this face, denoted by $d'$, is connected to both $u$ and $v$. Since $d$ and $d'$ as well as all the cutvertices of $\cal C$ are connected to both $u$ and $v$, by \cref{obs:triangles-at-u} and the fact that the vertices of $\cal C$ are on the outer face of $H' \setminus \{u,v\}$, it follows that all other vertices of $\cal C$ are connected to exactly one of $u$ or $v$.
Finally, if $d$ or $d'$ belonged to a block whose $B$-node in the $BC$-tree of $\cal C$ has degree $2$, then this block would contain three vertices all connected to both $u$ and $v$, which would imply the existence of a $K_5$ minor in $H'$ as in the proof of Property~$(\ref{ch:4})$.
\end{proof}

In view of Properties $(\ref{ch:2})$ and $(\ref{ch:4})$ of \cref{le:chain}, we refer to $\cal C$ as a \emph{chain of cycles} and to its blocks as \emph{cycles}, even when degenerated to single edges.
Moreover, we denote by $d'$ the non cutvertex of $\cal C$ different from $d$ that is incident to both $u$ and $v$, as specified in Property $(\ref{ch:5})$ of \cref{le:chain}.

\medskip\noindent\textbf{Making each chain of cycles skinny.} 
In order to incorporate the level-$\ell$ vertices that lie inside $\triangle$ into the morph of $\langle\Gamma_a(Q),\Gamma_b(Q)\rangle$, we perform a preliminary morph of $\Gamma_a(H')$ to a strictly convex drawing $\Gamma_a^s(H')$ of $H'$ that is \emph{skinny}, in the sense that it satisfies the following requirements with respect to the disk $D$ and the angle $\phi$ associated with the base edge $(u,v)$ derived from the morph of $\langle \Gamma_a(Q),\Gamma_b(Q) \rangle$ (see also \cref{fig:skinny} for an illustration).

\begin{enumerate}[R.1]
    \item \label{r:1} Every cycle of $\cal C$ is drawn inside the disk $D$.
    \item \label{r:2} Every cycle of $\cal C$ is drawn strictly convex.
    \item \label{r:3} The cutvertices of $\cal C$, as well as $d$ and $d'$, lie on the radius $r$ of $D$.
    \item \label{r:4} For every cycle of $\cal C$ and for every segment on its boundary, the smaller of the two angles formed at the intersection of the line through $r$ and the line through the segment is smaller than $\phi$. 
\end{enumerate}

\begin{figure}[t]
    \centering
    \includegraphics[scale=1.1,page=2]{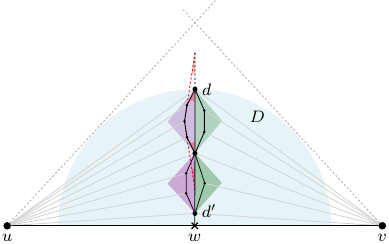}
    \caption{Illustration of the requirements R.\ref{r:1}--R.\ref{r:4} of a skinny drawing.}
    \label{fig:skinny}
\end{figure}

The existence of such a drawing is proven in the following lemma by means of a construction that exploits the properties of $\cal C$ given in \cref{le:chain}.  
\begin{lemma}\label{le:skinny}
There exists a drawing $\Gamma_a^s(H')$ of $H'$ that is strictly convex, skinny, and topologically equivalent to $\Gamma_a(H')$.
\end{lemma}
\begin{proof}
We prove the statement by construction.
Initially, we place $u$ and $v$ in the same positions as they are in $\Gamma_a(H')$. Further, we place the cutvertices of the chain of cycles $\cal C$ as well as $d$ and $d'$ on the radius $r$ in the order they appear in the chain to satisfy R.\ref{r:3}. For each cycle $c$ of $\cal C$, we proceed as follows. Let $x$ and $y$ be the two vertices of $c$ that have already been placed on $r$. Let $T^u_{c}$ and $T^v_{c}$ be two isosceles triangles sharing the same base $\overline{xy}$, such that the third vertex of each of them lies inside $D$ and on opposite sides of $r$ and such that the internal angles at $x$ and at $y$ are smaller than $\phi$; refer to the colored triangles in \cref{fig:skinny}. We place the vertices of $c$ that are incident only to $u$ (only to $v$) equidistant along a circular arc connecting $x$ and $y$ that lies completely inside $T^u_{c}$ (inside $T^v_{c}$, respectively).
By the definition of $T^u_{c}$ and $T^v_{c}$, and also by the fact that the two circular arcs are drawn completely inside $T^u_{c}$ and $T^v_{c}$, it follows that R.\ref{r:1}, R.\ref{r:2}, and R.\ref{r:4} are satisfied for the drawing of $c$. 

To complete the drawing of $\Gamma_a^s(H')$, we describe how to draw the subgraph $H'_c$ of $H'$ that is contained inside or on the boundary of $c$ such that every internal face of $H'_c$ is strictly convex.
Since $H'_c$ is drawn convex in $\Gamma_a(H')$, it admits a strictly convex drawing for any given strictly convex drawing of its outer face~\cite{chiba1984linear}. Thus, we can apply the algorithm in~\cite{chiba1984linear}  to construct a strictly convex drawing of $H'_c$, whose outerface is the drawing of $c$ satisfying R.\ref{r:1}-R.\ref{r:4}.
Finally, we add the edges incident to $u$ and $v$ that are contained inside $\triangle$ to the resulting drawing, which does not introduce crossings due to R.\ref{r:4}. This completes the construction of $\Gamma_a^s(H')$. Since every cycle in $\cal C$ satisfies R.\ref{r:1}--R.\ref{r:4} and since by \cref{obs:triangles-at-u} all faces incident to $u$ and $v$ in $H'$ are triangular, the drawing $\Gamma_a^s(H')$ is strictly convex and skinny as desired. Since our construction and the algorithm in~\cite{chiba1984linear} maintain the cyclic order of the edges around each vertex, we have that $\Gamma_a^s(H')$ is topologically equivalent to $\Gamma_a(H')$. This concludes the proof.
\end{proof}

To describe the morph between $\Gamma_a(H')$ and $\Gamma_a^s(H')$, we need some more work.
Since both drawings are strictly convex and topologically equivalent, the preconditions of \cref{th:convex} are met. However, to ensure that this morph can be done independently for each \piece, we further need to guarantee that vertices $u$ and $v$ do not move and that all vertices of $H'$ remain inside $\triangle$ throughout the morph. As stated in \cref{th:convex}, this can be achieved if the (triangular) outer face is drawn the same in the two input drawings, which is not necessarily the case for $\Gamma_a( H')$ and $\Gamma_a^s(H')$ because of the position of $d$ (recall that $u$ and $v$ have the same position in $\Gamma_a(H')$ and $\Gamma_a^s(H')$).
To this end, we augment $\Gamma_a(H')$ and $\Gamma_a^s(H')$ by adding a new vertex $d^*$ in the outer face of $H'$ and connect it to $u$, $v$, and $d$. Moreover, we place $d^*$ at the same position inside $\triangle$ in both $\Gamma_a(H')$ and $\Gamma_a^s(H')$ so that the triangle formed by $u$, $v$, and $d^*$ contains all the other vertices of $H'$ (in particular, $d$). The edge $(d,d^*)$ can always be drawn without crossings, as $u$, $v$, and $d$ were the vertices on the outer face of $H'$ before. After this augmentation, we apply \cref{th:convex} to compute the desired morph of $\langle \Gamma_a(H'),\Gamma_a^s(H') \rangle$, and then we remove $d^*$ from the drawings.

\medskip\noindent\textbf{Performing the global morph.}
Applying the above procedure for each \piece yields a drawing $\Gamma_a(G')$ of the supergraph $G'$ of $G$ that is the union of $Q$ and all the graphs $H_\triangle'$ corresponding to every \piece $\triangle$.
Observe that $\Gamma_a(G')$ is composed of $\Gamma_a(Q)$ and the skinny drawing $\Gamma^s_a(H_\triangle')$ of every graph $H_\triangle'$.
To perform the global morph, recall that the vertices of the subgraph $Q$ of $G'$ follow the same trajectories as in the morph between $\Gamma_a(Q)$ and $\Gamma_b(Q)$ (which has been recursively computed). The level-$\ell$ vertices of each subgraph $H_\triangle'$ are moved inside $\triangle$, which again ensures that this can be done independently for each \piece. In the following we describe the trajectories of one such subgraph.
We denote this subgraph as $H'$ and adopt the same notation as before.

Since the trajectories of $u$ and $v$ are specified by the morph between $\Gamma_a(Q)$ and $\Gamma_b(Q)$, we only describe the trajectories of $H' \setminus \{u,v\}$, i.e., the vertices of level $\ell$; see \cref{fig:traj} for an example. 
The drawing of $H' \setminus \{u,v\}$ is a copy of $\Gamma_a^s(H' \setminus \{u,v\})$ rotated and translated so that the cutvertices of $\cal C$ as well as $d$ and $d'$ lie on the radius of $D$ perpendicular to $(u,v)$, and the distance between $w$ and $d'$ is the same as in $\Gamma_a^s(H')$. 
This ensures that the drawing of $H'$ remains skinny, planar (by R.\ref{r:4}), and strictly convex at every time instant.

\begin{figure}[t]
    \centering
    \includegraphics[scale=1,page=3]{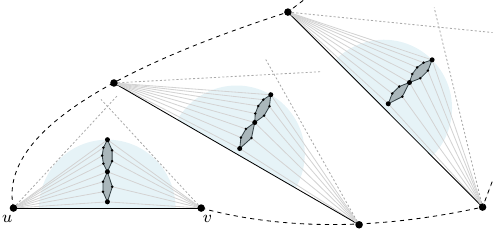}
    \caption{Computing the trajectories for the vertices of $H'\setminus\{u,v\}$ based on the, already computed, trajectories of $u$ and $v$.}
    \label{fig:traj}
\end{figure}

Let $\Gamma_b(G')$ be the drawing of $G'$ obtained so far. 
The next step of the morph is to transform, for each subgraph $H_\triangle'$, the current skinny drawing  $\Gamma_b^s(H_\triangle')$ in $\Gamma_b(G')$ to $\Gamma_b(H_\triangle')$. By construction, $\Gamma_b^s(H_\triangle')$ and $\Gamma_b(H_\triangle')$ are topologically equivalent and strictly convex. Similarly as for $\Gamma_a(H_\triangle')$, we insert vertex $d^*$ so that the outer face of $H_\triangle'$ is drawn the same in both $\Gamma^s_b(H'_\triangle)$ and $\Gamma_b(H'_\triangle)$, which allows to apply \cref{th:convex} independently for each $H_\triangle'$.
The target drawing $\Gamma_b(G)$ is obtained by removing the vertices and edges in $G' \setminus G$ and by reinserting the crossed edges in the marked faces. This concludes the proof of \cref{thm:main}.

\section{Implications of \cref{thm:main}}\label{se:implications}

In this section, we discuss the applicability of \cref{thm:main} by presenting meaningful families of $1$-planar graphs that admit \good drawings. 

An $n$-vertex $1$-planar graph has at most $4n-8$ edges~\cite{BSW83}, and if it achieves exactly this density, then it is called \emph{optimal}. Moreover, any $1$-planar drawing of an optimal $1$-planar graph $G$ is such that the uncrossed edges induce a plane triconnected quadrangulation $P$, while each pair of crossing edges of $G$ is drawn inside a corresponding face of $P$~\cite{S10b}. When restricting to straight-line drawable $1$-planar graphs, this bound is reduced to $4n-9$~\cite{DBLP:journals/ipl/Didimo13}. Similarly to the general case, an optimal $1$-planar straight-line drawing is one in which the uncrossed edges induce a plane triconnected graph whose every inner face is a quadrangle, while the outer face is a triangle~\cite{DBLP:journals/ipl/Didimo13}. As a consequence, we obtain that each kite is empty and its kite edges are uncrossed. Therefore, any optimal $1$-planar straight-line drawing is \good.  

Another family of $1$-planar graphs that recently attracted considerable attention are the \emph{IC-planar} graphs \cite{AMC10,DBLP:journals/tcs/BrandenburgDEKL16,CzapS17,LiottaM16}, which admit $1$-planar drawings where the crossed edges induce a matching. Note that both the binding edges and the kite edges that are part of an IC-planar drawing are uncrossed. However, P.\ref{pr:p2} of \cref{def:good-drawing} may not be satisfied for those kite edges that are not part of the drawing. It follows that, if an IC-planar drawing is also \emph{kite-augmented}~\cite{DBLP:journals/algorithmica/Brandenburg19}, i.e., it contains all kite edges, then it is \good.

Overall, the following result is a corollary of \cref{thm:main}.

\begin{corollary}\label{co:ic-planar}
There exists a morph between any pair of topologically-equivalent  optimal $1$-planar or kite-augmented IC-planar straight-line drawings.
\end{corollary}

We conclude this section with a remark. As already mentioned, Chambers et al.~\cite{morphing-torus:soda2021} studied morphs of toroidal graphs and asked to generalize their result to surfaces of higher genus. We note that, since an $n$-vertex graph embeddable on a surface of genus $g$ has at most $3n+6(g-1)$ edges, while $n$-vertex optimal $1$-planar straight-line drawable graphs have $4n-9$ edges, it follows that the latter do not admit an embedding (without edge crossings) on any surface of bounded genus. Thus, a solution to the open problem by Chambers et al. would not provide morphs of \good drawings.

\section{Open Problems}\label{sec:conclusions}
We made a first step towards the problem of morphing pairs of non-planar drawings. Besides the general open problem of morphing any two such drawings~\cite{DBLP:conf/icalp/AngeliniLBFPR14}, the main questions that stem from our research are as follows: 
\begin{inparaenum}[(i)]
\item Is it possible to compute morphs of \good drawings where the vertex trajectories have bounded complexity?
\item Regardless of the complexity, is it possible to extend \cref{thm:main} to all $1$-planar drawings or to meaningful families of $k$-planar drawings ($k>1$)?
\item Further families of beyond-planar graphs~\cite{DBLP:journals/csur/DidimoLM19} could also be considered, for instance, does every pair of RAC drawings~\cite{DBLP:books/sp/20/Didimo20} admit a morph?
\end{inparaenum}

\bibliographystyle{splncs04}
\bibliography{bibliography.bib}

\end{document}